\documentclass{sig-alternate}
\usepackage[utf8]{inputenc}
\usepackage{stmaryrd}
\usepackage{mathrsfs}
\usepackage{hyperref}

\usepackage{amsthm}

\usepackage{tikz}
\usetikzlibrary{calc, matrix, backgrounds,shapes.geometric}

\usepackage[ruled, vlined]{algorithm2e}

\newfont{\mycrnotice}{ptmr8t at 7pt}
\newfont{\myconfname}{ptmri8t at 7pt}
\let\crnotice\mycrnotice%
\let\confname\myconfname%

\permission{Permission to make digital or hard copies of all or part
  of this work for personal or classroom use is granted without fee
  provided that copies are not made or distributed for profit or
  commercial advantage and that copies bear this notice and the full
  citation on the first page. Copyrights for components of this work
  owned by others than ACM must be honored. Abstracting with credit is
  permitted. To copy otherwise, or republish, to post on servers or to
  redistribute to lists, requires prior specific permission and/or a
  fee. Request permissions from permissions@acm.org.}

\conferenceinfo{ISSAC'16,}{July , 2016, London, Canada.\\
{\mycrnotice{Copyright is held by the owner/author(s). Publication rights licensed to ACM.}}}
\copyrightetc{ACM \the\acmcopyr}
\crdata{9\ ...\$15.00.\\
DOI: http://dx.doi.org/}

\begin{document}

\ResetInOut{I}

\newtheorem{theo}{Theorem}[section]
\newtheorem{lem}[theo]{Lemma}
\newtheorem{prop}[theo]{Proposition}
\newtheorem{cor}[theo]{Corollary}
\newtheorem{conj}{Conjecture}
\newtheorem{quest}[theo]{Question}

\theoremstyle{definition}
\newtheorem{defn}[theo]{Definition}
\newtheorem{expl}[theo]{Example}

\theoremstyle{remark}
\newtheorem{rmk}[theo]{Remark}
\newtheorem{note}[theo]{Note}

%


\title{On the p-adic stability of the FGLM algorithm}

%
%
%
%
%

\numberofauthors{2}
%
%
\author{\alignauthor
Guénaël Renault\\
       \affaddr{PolSys project INRIA Paris-Rocquencourt,}
       \affaddr{UPMC Univ. Paris 06, CNRS, UMR 7606, LIP6,}
       \affaddr{Paris, France}
       \email{guenael.renault@lip6.fr}
\alignauthor
Tristan Vaccon\\
       \affaddr{JSPS--Rikkyo University} \\
       \affaddr{Tokyo, Japan}
       \email{vaccon@rikkyo.ac.jp}
}

\maketitle
\begin{abstract} 

Nowadays, many strategies to solve polynomial systems use the computation of a Gröbner basis for the graded reverse lexicographical ordering, followed by a change of ordering algorithm to obtain a Gröbner basis for the lexicographical ordering.
The change of ordering algorithm is crucial for these strategies.
We study the $p$-adic stability of the main change of ordering algorithm, FGLM.

We show that FGLM is stable and give explicit upper bound on the loss of precision occuring in its execution.
The variant of FGLM designed to pass from the grevlex ordering to a Gröbner basis in shape position 
is also stable.

Our study relies on the application of Smith Normal Form computations for linear algebra.


\end{abstract}




%

\section{Introduction}

The advent of arithmetic geometry has seen the emergence of questions that are
purely local (\textit{i.e.} where a prime $p$ is fixed at the very beginning and
one can not vary it). As an example, one can cite the work of Caruso and Lubicz
\cite{Caruso:2014} who gave an algorithm to compute lattices in some $p$-adic
Galois representations. A related question is the study of $p$-adic deformation
spaces of Galois representations. Since the work of Taylor and Wiles
\cite{Taylor:1995}, one knows that these spaces play a crucial role in many
questions in number theory. Being able to compute such spaces appears then as an
interesting question of experimental mathematics and require the use of purely
$p$-adic Gröbner bases and more generally $p$-adic polynomial system solving.

Since \cite{Vaccon:2014}, it is possible to compute a Gröbner basis, under some
genericness assumptions, for a monomial ordering $\omega$ of an ideal generated
by a polynomial sequence $F=(f_1,\dots,f_s) \subset \mathbb{Q}_p [X_1, \dots,
  X_n]$ if the coefficients of the $f_i$'s are given with enough initial precision.
Unfortunately, one of the genericness assumptions (namely, the sequence
$(f_1,\dots,f_i)$ has to be weakly-$\omega$) is at most generic for the graduate
reverse lexicographical (denoted grevlex in the sequel) ordering (conjecture of
Moreno-Socias). Moreover, in the case of the lexicographical ordering (denoted
lex in the sequel) this statement is proved to be generically not satisfied for
some choices of degrees of the entry polynomials. In the context of polynomial
system solving, where the lex plays an important role, this fact becomes a
challenging problem that is essential to overcome.

Thus, in this paper, we focus on the fundamental problem of change of ordering
for a given $p$-adic Gröbner basis. In particular we provide precise results in
the case where the input basis has a grevlex ordering and one wants to compute
the lex basis corresponding. We will use the following notations.

\subsection{Notations}
Throughout this paper, $K$ is a field with a discrete valuation $val$ such that
$K$ is complete with respect to the norm defined by $val$. We denote by $R=O_K$
its ring of integers, $m_K$ its maximal ideal and $k=O_K/m_K$ its fraction
field. We denote by CDVF (complete discrete-valuation field) such a field. We
refer to Serre's Local Fields \cite{Serre:1979} for an introduction to such
fields. Let $\pi \in R$ be a uniformizer for $K$ and let $S_K \subset R$ be a
system of representatives of $k=O_K/m_K.$ All numbers of $K$ can be written
uniquely under its $\pi$-adic power series development form: $\sum_{k \geq l}
a_k \pi^l$ for some $ l \in \mathbb{Z}$, $a_k \in S_K$.

The case that we are interested in is when $K$ might not be an effective field,
but $k$ is (\textit{i.e.} there are constructive procedures for performing
rational operations in $k$ and for deciding whether or not two elements in $k$
are equal). Symbolic computation can then be performed on truncation of
$\pi$-adic power series development. We will denote by finite-precision CDVF
such a field, and finite-precision CDVR for its ring of integers. Classical
examples of such CDVF are $K = \mathbb{Q}_p$, with $p$-adic valuation, and
$\mathbb{Q}[[X]]$ or $\mathbb{F}_q[[X]]$ with $X$-adic valuation. We assume that
$K$ is such a finite-precision CDVF.

The polynomial ring $K[X_1,\dots, X_n]$ will be denoted $A$, and
$u=(u_1,\dots,u_n) \in \mathbb{Z}_{\geq 0}^n$, we write $x^u$ for
$X_1^{u_1} \dots X_n^{u_n}.$

\subsection{Mains results}

In the context of $p$-adic algorithmic, one of the most important
behavior to study is the stability of computation: how the quality of
the result, in terms of $p$-adic precision, evolves from the input. To
quantify such a quality, it is usual to use an invariant, called
condition number, related to the computation under study. Thus, we
define such an invariant for the change of ordering.

\begin{defn} \label{defn cond FGLM} Let $I \subset A$ be a
  zero-dimensional ideal. Let $\leq_1$ and $\leq_2$ be two monomial
  orderings on $A$. Let $B_{\leq_1}$ and $B_{\leq_2}$ be the canonical
  bases of $A/I$ for $\leq_1 $ and $\leq_2$. Let $M$ be the matrix whose
  columns are the $NF_{\leq_1} (x^\beta)$ for $x^\beta \in B_{\leq_2}$. We
  define the condition number of $I$ for $\leq_1$ to $\leq_2$, with
  notation $cond_{\leq_1, \leq_2}(I)$ (or $cond_{\leq_1, \leq_2}$ when
  there is no ambiguity) as the biggest valuation of an invariant
  factor of $M$.
\end{defn}

We can now state our main result on change of ordering of $p$-adic
Gröbner basis.

\begin{theo} \label{thm stabilite fglm} Let $\leq_1$ and $\leq_2$ be two
  monomial orderings. Let $G=(g_1,\dots,g_t) \in K[X_1,\dots, X_n]^t$ be
  an approximate reduced Gröbner basis for $\leq_1$ of the ideal $I$
  it generates, with $\dim I = 0$ and $\deg I = \delta$, and with
  coefficients known up to precision $O(\pi^N)$. 
  Let $\beta$ be the smallest valuation of a coefficient in $G.$    
  Then, if   $N>cond_{\leq_1, \leq_2}(I)$, the stabilized FGLM 
  Algorithm, Algorithm \ref{FGLM stabilise}, computes a Gröbner basis $G_2$ 
  of $I$ for $\leq_2$. The   coefficients of the polynomials of 
  $G_2$ are known up to precision
  $N+n^2(\delta+1)^2 \beta-2cond_{\leq_1, \leq_2}$. 
  The time-complexity is in $O(n \delta^3)$.
\end{theo}

In the case of a change of ordering from grevlex to lex, we provide a more
precise complexity result:

\begin{theo} \label{thm:FGLMp:grevlex:lex} With the same notations and
  hypothesis as in Theorem~\ref{thm stabilite fglm}. If $\leq_1$, $\leq_2$ are
  respectively instantiated to grevlex and lex, and if we assume that the ideal
  $I$ is in shape position.
  Then, the adapted FGLM Algorithm for general position, Algorithm 
  \ref{FGLM stabilise avec shape depuis grevlex}, computes a
  Gröbner basis $G_2$ of $I$ for lex, in shape position. The
  coefficients of the polynomials of $G_2$ are known up to precision
  $N+\beta \delta-2cond_{\leq_1, \leq_2}$. The time-complexity is in
  $O(n\delta^2)+O(\delta^3)$.
\end{theo}

In order to obtain these results, one has to tackle technical problems related
to the core of the FGLM algorithm. Thus, we first present a summary of some
important facts on this algorithm. Then we present more precisely the underlying
problems in the $p$-adic situation.

\subsection{The FGLM algorithm}

For a given zero-dimensional $I$ in a polynomial ring ${A}$, the FGLM
algorithm \cite{Faugere:1993} is mainly based on computational linear algebra
in ${A}/I$. It allows to compute a Gröbner basis $G_2$ of $I$ for a
monomial ordering $\leq_2$ starting from a Gröbner basis $G_1$ of $I$ for a
first monomial ordering $\leq_1$. To solve polynomial systems, one possible
method is the computation of a Gröbner basis for lex. However, computing a
Gröbner basis for lex by a direct approach is usually very time-consuming. The
main application of the FGLM algorithm is to allow the computation of a Gröbner
basis for lex by computing a Gröbner basis for grevlex then by applying a change
of ordering to lex. The superiority of this approach is mainly due to the fact
that the degrees of the intermediate objects are well controlled during the
computation of the grevlex Gröbner basis. The second step of this general method
for polynomial system solving, is what we call the FGLM algorithm. Many variants
and improvements (in special cases) of the FGLM algorithm have been published,
\textit{e.g.} Faugère and Mou in
\cite{Faugere-Mou:2011,Faugere-Mou:2013,Mou:2013} and Faugère, Gaudry, Huot and
Renault \cite{Faugere-Huot:2013,Faugere:2014,Huot:13} takes advantages of sparse
linear algebra and fast algorithm in linear algebra to obtain efficient
algorithms. In this paper, as a first study of the problem of loss of precision
in a change of ordering algorithm, we follow the original algorithm. This study
already brings to light some problems for the loss of precision and propose
solutions to overcome them. Thus, the FGLM algorithm we consider can be sketched
as follows:

\begin{enumerate}
\item Order the images in ${A}/I$ of the monomials of ${A}$
  according to $\leq_2$.
\item \label{intro:step:independence} Starting from the first monomial, test the
  linear independence in ${A}/I$ of a monomial $x^\alpha$ with the
  $x^\beta$ smaller than it for $\leq_2$.
\item In case of independence, $x^\alpha$ is added to the canonical (for
  $\leq_2$) basis ${A}/I$ in construction.
\item \label{intro:step:relation} Otherwise, $x^\alpha \in LM(I)$ and the linear
  relation with the $x^\beta$ smaller than it for $\leq_2$ give rise to a
  polynomial in $I$ whose leading term is $x^\alpha.$
\end{enumerate}

Precision problems arise in step \ref{intro:step:independence} and
\ref{intro:step:relation}. The first one is the issue of testing the
independence of a vector from a linear subspace. While it is possible to prove
independence when the precision is enough, it is usually not possible to prove
directly dependence. It is however possible to prove some dependence when there
is more vectors in a vector space than the dimension of this vector space. It is
indeed some dimension argument that permits to prove the stability of FGLM. We
show (see Section \ref{section: stabilite FGLM}) that it is enough to treat
approximate linearly dependence (up to some precision) in the same way as in
the non-approximate case and to check at the end of the execution of the
algorithm that the number of independent monomials found is the same as the
degree of the ideal. The second issue corresponds to adding the computation of
an approximate relation. We show that the same idea of taking approximate linear
dependence as non-approximate and check at the end of the computation is enough.

\subsection{Linear algebra and Smith Normal Form}

As we have seen, the FGLM algorithm relies mainly on computational linear
algebra: testing linear independence and solving linear systems.

The framework of differential precision of \cite{Caruso:2014:2} has been applied
to linear algebra in \cite{Caruso:2015} for some optimal results on the behavior
of the precision for basic operations in linear algebra (matrix multiplication,
LU factorization). From this analysis it seems clear, and this idea is well
accepted by the community of computation over $p$-adics, that using the Smith
Normal Form (SNF) to compute the inverse of a matrix or to solve a well-posed
linear system is highly efficient and easy to handle. Moreover, it always
achieves a better behavior than classical Gaussian elimination, even allowing
gain in precision in some cases. Its optimality remains to be proved but in
comparison with classical Gaussian elimination, the loss of precision is far
fewer. \footnote{See Chapter 1 of \cite{Vaccon-these} for more details on the
  comparison between these two strategies.}

This is the reason why we use the SNF in the $p$-adic version of FGLM we propose
in this paper. In Section \ref{sec:SNF} we briefly recall some properties of the
SNF and its computation. We also provide a dedicated version of SNF computation
for the FGLM algorithm. More precisely, to apply the SNF computation to
iterative tests of linear independence (as in step
\ref{intro:step:independence}), we adapt SNF computation into some iterative SNF
in Algorithm \ref{Update FGLM stabilise}. This allows us to preserve an overall
complexity in $O(n \delta^3).$


\section{SNF and linear systems}
\label{sec:SNF}

\subsection{SNF and approximate SNF}

We begin by presenting our main tool, the SNF of a matrix in $M_{n,m}(K)$:

\begin{prop}
Let $M \in M_{n,m}(K)$. There exist some $P \in GL_n (O_K)$, $val(\det P)=0 $, $Q \in GL_m (O_K)$, $\det Q = \pm 1$ and $ \Delta \in M_{n,m}(K)$ such that $M = P \Delta Q$ and $\Delta$ is diagonal, with diagonal coefficients being $\pi^{a_1},\dots,\pi^{a_s},0,\dots,0$ with $a_1 \leq \cdots \leq a_s$ in $\mathbb{Z}.$
$\Delta$ is unique and called the Smith Normal Form of $M$, and we say that $P,\Delta,Q$ realize the SNF of $M$. The $a_i$ are called the invariant factors of $M.$
\end{prop}

In a finite-precision context, we introduce the following variant of the notion of SNF:

\begin{defn} \label{FNS approchee}
Let $M \in M_{n,m}(K)$, known up to precision $O(\pi^l)$. We define an  \textbf{approximate SNF} of $M$ as a factorization \[M = P \Delta Q \] with $P \in M_n(R)$, $val(\det P)=0,$ $Q \in M_m(R)$ with $\det Q=\pm 1$ known up to precision $O(\pi^l)$ and $ \Delta \in M_{n,m}(K)$ such that $\Delta = \Delta_0+O(\pi^l)$, where $ \Delta_0 \in M_{n,m}(K)$ is a diagonal matrix, with diagonal coefficients of the form $\Delta_0 [1,1]= \pi^{\alpha_1}, \dots, \Delta_0 [\min (n,m),\min (n,m)] = \pi^{\alpha_{\min (n,m)}}$ with $\alpha_1 \leq \dots \leq \alpha_{\min (n,m)}$. $\alpha_i = + \infty$ is allowed. $(P, \Delta, Q)$ are said to realize an approximate SNF of $M.$
\end{defn}

To compute an approximate SNF, with use Algorithm \ref{algo approx snf}.
\vspace{-0.5cm}
\begin{algorithm}

  \scriptsize
 \caption{SNFApproximate:  Computation of an approximate SNF}
 \SetKwInOut{Input}{Input}\SetKwInOut{Output}{Output}

 \Input{$M \in M_{n \times m}(K)$, known up to precision $O(\pi^l),$ with $l>cond(M)$}
 \Output{$P,\Delta,Q$ realizing an approximate SNF of $M$.} 

\BlankLine

Find $i,j$ such that the coefficient $M_{i,j}$ realize $\min_{k,l} val (M_{k,l})$ \;
Track the following operations to obtain $P$ and $Q$ \;
Swap rows $1$ and $i$ and columns $1$ and $j$ \;
Normalize $M_{1,1}$ to the form $\pi^{a_1}+O(\pi^l)$ \; 
By pivoting reduce coefficients $M_{i,1}$ ($i>1$) and $M_{1,j}$ ($j>1$) to $O(\pi^l).$ \;
\textbf{Recursively}, proceed with $M_{i \geq 2, j \geq 2}$ \;
Return $P,M, Q$ \; \label{algo approx snf}
\end{algorithm}
\vspace{-0.5cm}

Behaviour of Algorithm \ref{algo approx snf} is given by the following proposition:

\begin{prop}\label{prop snf approchee}
Given an input matrix $M$, of size $n \times m$, with precision
$(O(\pi^l)$ on its coefficients, Algorithm \ref{algo approx snf}
terminates and returns $U,\Delta, V$ realising an approximate SNF of
$M$. Coefficients of $U, \Delta $ and $V$ are known up to precision
$O(\pi^l)$. Time complexity is in $O( \min(n,m) \max (n,m)^2)$
operations in $K$ at precision $O(\pi^l)$.
\end{prop}

Now, it is possible to compute the SNF of $M$, along with an
approximation of a realization, from some approximate SNF of $M$ with Algorithm \ref{algo snf precisee}.

\vspace{-0.5cm}
\begin{algorithm}
  \scriptsize
 \caption{SNFPrecised : from approximate SNF to SNF}
 \SetKwInOut{Input}{Input}\SetKwInOut{Output}{Output}

 \Input{$(U, \Delta,V)$ (precision $O(\pi^l)$) realizing an approximate SNF of $M\in M_{n \times m}(K),$ of full rank. We assume $cond(M)<l$.}
 \Output{$\Delta_0$ the SNF of $M$, and $U',V'$ known with precision $O(\pi^{l-cond(M)})$ such that $M=U' \Delta_0 V'$, $val(\det U')=0$ and $\det V = \pm 1.$} 

$\Delta_0 \leftarrow \Delta$  \;
Track the following operations to obtain $P$ and $Q$ \;
$t := min(n,m)$ \; 
\For{$i$ from $1$ to $t$}{
Normalize $\Delta_0[i,i]$ \;
\eIf{$t=m$}{
By pivoting with $\Delta_0[i,i]$, eliminate the coefficients $\Delta_0[j,i]$ \;}
{By pivoting with $\Delta_0[i,i]$, eliminate the coefficients $\Delta_0[i,j]$ \;}
}
Return $\Delta_0,P,Q$ \;

\label{algo snf precisee}
\end{algorithm}
\vspace{-0.5cm}

\begin{prop}
Given an input matrix $M$, of size $n \times m$, with precision
$(O(\pi^l)$ on its coefficients ($l>cond(M)$), and $(U, \Delta,V)$,
known at precision $O(\pi^l)$, realizing an approximate SNF of $M$,
Algorithm \ref{algo snf precisee} computes the SNF of $M$, with $U'$
and $V'$ known up to precision $O(\pi^{l-cond(M)})$. Time-complexity
is in $O( \max (n,m)^2).$
\end{prop}
We refer to \cite{Vaccon:2014, Vaccon-these} for more details on how to
prove this result. We can then conclude on the computation of the SNF:

\begin{theo}
Given an input matrix $M$, of size $n \times m$, with precision
$O(\pi^l)$ on its coefficients ($l>cond(M)$), then by applying
Algorithms \ref{algo approx snf} and \ref{algo snf precisee}, we
compute $P,Q,\Delta$ with $M=P \Delta Q$ and $\Delta$ the SNF of
$M$. Coefficients of $P$ and $Q$ are known at precision
$O(\pi^{l-cond(M)})$. Time complexity is in $O(\max(n,m)^2 \min(n,m))$
operations at precision $O(\pi^l)$.
\end{theo}

\subsection{Solving linear systems}

Computation of $P$ and $Q$ in the previous algorithms can be slightly
modified to obtain (approximation of) $P^{-1}$ and $Q^{-1},$ and thus
$M^{-1}.$

\begin{prop} \label{prop snf inverses}
Using the same context as the previous theorem, by modifying
Algorithms \ref{algo approx snf} and \ref{algo snf precisee} using the
inverse operations of the one to compute $P$ and $Q$, we can obtain
$P^{-1}$ and $Q^{-1}$ with precision $O(\pi^{l-cond(M)})$.  When $M
\in GL_n(K),$ using $M^{-1}=Q^{-1} \Delta^{-1} P^{-1},$ we
get $M^{-1}$ with precision $O(\pi^{l-2cond(M)})$.  Time complexity is
in $O(n^3)$ operations at precision $O(\pi^{l})$.
\end{prop}

We can then estimate the loss in precision in solving a
linear system:

\begin{theo} \label{thm sys lin}
Let $M \in GL_n(K)$ be a matrix with coefficients known up
to precision $O(\pi^l)$ with $l>2cond(M).$ Let $Y \in K^n$
be known up to precision $O(\pi^l)$.  Then one can solve $Y=MX$ in
$O(n^3)$ operations at precision $O(\pi^l).$ 
$X$ is known at precision
$O(\pi^{l-2cond(M)})$.
\end{theo}

When the system is not square but we can ensure that $Y \in Im(M)$,
then we have the following variant:

\begin{prop} \label{prop sys lin pas carre}
Let $M \in M_{n,m}(K)$ be full rank matrix, with
coefficients known at precision $O(\pi^l),$ with $l> 2cond(M)$. Let $Y
\in K^n$ known at precision $O(\pi^l)$ be such that $Y \in
Im(M)$.  Then, we can compute $X$ such that $Y=MX$, with precision
$O(\pi^{l-2cond(M)})$ and time-complexity \penalty-3000 $O(nm \max (n,m))$ operations 
in $K$ at precision $O(\pi^l).$ 
\end{prop}
\vspace{-0.5cm}
\begin{algorithm}

  \scriptsize
 \SetKwInOut{Input}{Input}\SetKwInOut{Output}{Output}

 \Input{ The reduced Gröbner basis $G$ of the zero-dimensional ideal $I \subset A$ for a monomial ordering $\leq$. $\deg I = \delta$. $B_\leq =(1=\epsilon_1 \leq  \epsilon_2 \leq \dots \leq \epsilon_\delta)$ the canonical basis of $A/I$ for $\leq.$\\
A monomial ordering $\leq_2$.}
 \Output{An approximate Gröbner basis $G_2$ of $I$ for $\leq_2$, or \textbf{Error} if the precision is not enough.}

Compute the multiplication matrices $T_1,\dots,T_n$ for $I$ and $\leq$ with Algorithm \ref{calcul des matrices de multiplication} \;
$B_2 := \{ 1 \}$ ; $\mathbf{v}= [{}^t (1,\dots,0)]$ ; $G_2 := \emptyset$ \;
$L := \{ (1,n),(1,n-1),\dots,(1,1) \}$ \;
$Q1,Q2,P1,P2,\Delta:=I_1,I_1,I_\delta,I_\delta,\mathbf{v}$ \;
\While{ $L \neq \emptyset$ }{
$m := L[1]$ ; Erase $m$ of $L$ \;
$j:=m[1]$ ; $i:=m[2]$ \;
$v := T_i \mathbf{v}[j]$ \;
$s:= card(B_2)$ \;
$\lambda = {}^t (\lambda_1,\dots,\lambda_\delta) := P_1 v$ \;
\eIf{ we have no significant digit on $\lambda_{s+1},\dots,\lambda_\delta$ (\textit{i.e.} those are $O(\pi^v)$)}{
Compute the SNF of $\mathbf{v}$ from the approximate SNF given by $\Delta,$ $P_1,Q_1$ and their inverses $P_2,Q_2$, with Algorithm \ref{algo snf precisee} \;
Find $W$ such that $\mathbf{v}W=v$, thanks to the SNF of $\mathbf{v}$, assuming that $v \in Im(\mathbf{v})$ (Proposition \ref{prop sys lin pas carre})\;
$G_2 :=  G_2 \cup \{B_2[j] x_i-\sum_{l=1}^s W_l B_2[l] \}$}
{$B_2:=B_2 \cup \{B_2[j] x_i \}$ \;
$\mathbf{v}=\mathbf{v} \cup [v ]$ \;
$L:=IncreasingSort ( L \cup [(s+1,l) \vert 1\leq l \leq n] , \leq_2 )$ \;
Remove the repeats in $L$ \;
$Update(\mathbf{v},s,P_1,P_2,Q_1,Q_2,\Delta)$ \;}
Remove from $L$ all the multiples of $LM_{\leq_2}(G_2)$ \;
}
\eIf{$card(B_2)=\delta$}{
Return $G_2$ \;
}
{Return "\textbf{Error}, not enough precision"}

 \caption{Stabilized FGLM} \label{FGLM stabilise}
\end{algorithm}
\vspace{-0.5cm}
\normalsize 
\section{Stability of FGLM}
\label{section: stabilite FGLM}
\subsection{A stabilized algorithm}
This section is devoted to the study of the FGLM algorithm at finite precision over $K$.
More precisely, we provide a stable adaptation of this algorithm.
The main difference with the classical FGLM algorithm consists in the 
replacement of the row-echelon form computations by SNF computation, as in Section \ref{sec:SNF}.
This way, we are able to take advantage of the smaller loss in precision of the SNF,
and the nicer estimation on the behaviour of the precision it yields.

FGLM is made of Algorithms \ref{calcul des matrices de multiplication}, \ref{FGLM stabilise}
and \ref{Update FGLM stabilise}, with Algorithm \ref{FGLM stabilise} being the 
main algorithm.


\begin{rmk} For the linear systems solving in Algorithm \ref{FGLM stabilise}, we use the computation of a SNF from an approximate SNF thanks to Algorithm \ref{algo snf precisee}, and then solve the system as in \ref{prop sys lin pas carre}.
\end{rmk}

The remaining of this Section is devoted to the proof of our main theorem \ref{thm stabilite fglm}.

\subsection{Proof of the algorithm}

To prove Theorem \ref{thm stabilite fglm} regarding the stability of Algorithm \ref{FGLM stabilise}, we first begin by a lemma to control the behaviour of the condition number of $\mathbf{v}$ during the execution of the algorithm, and then apply it to prove each component of the proof one after the other.

A preliminary remark can be given: over infinite precision, correction and termination of Algorithm \ref{FGLM stabilise} are already known. Indeed, the only difference in that case with the classical FGLM algorithm is that the independence testing and linear system solving are done using (iterated) SNF instead of reduced row-echelon form computation. 

\subsubsection{Growth of the condition in iterated SNF}

In order to control the condition number of $\mathbf{v}$ during the execution of the algorithm, and thus control the precision, we use the following lemma:

\begin{lem} \label{lem snf iteree}
Let $M \in M_{s, \delta}(K)$ be a matrix, with $s < \delta$ being integers. Let $v \in K^\delta$ be a vector and $M' \in M_{s+1, \delta}(K)$ the matrix obtained by adjoining the vector $v$ as an $(s+1)$-th column for $M.$ 
Let $c=cond(M)$, and $c'=cond(M').$ We assume $c,c' \neq + \infty$ (\textit{i.e.}, the matrices are of full-rank). Then $c \leq c'$.
\end{lem}
\begin{proof}
We use the following classical fact : let $d_{s}'$ be the smallest valuation achieved by an $s \times s$ minor of $M'$, and $d_{s+1}'$ the smallest valuation achieved by an $(s+1) \times (s+1)$ minor of $M'$, then\footnote{This is a direct consequence of the fact that for an ideal  $\mathscr{I}$ in the ring of integers of a discrete valuation field, any element $x \in \mathscr{I}$ such that $val(x)=\min val(\mathscr{I})$ generates $\mathscr{I}$, with the converse being true.}  $c'=d_{s+1}'-d_{s}'$.

In our case, let $P,Q,\Delta$ be such that $\Delta$ is the SNF of $M,$ $P \in GL_{\delta}(R),$ $Q \in GL_s (R)$ and $PMQ=\Delta$. Then, by augmenting trivially $Q$ to get $Q'$ with $Q'_{s+1,s+1}=1$, we can write:
\vspace{-0.5cm}
\scriptsize
\[PM'Q'= \begin{tikzpicture}[baseline=(current bounding box.center),scale=0.2]
\matrix (m) [matrix of math nodes,nodes in empty cells,right delimiter={]},left delimiter={[}, ampersand replacement=\& ]{
\pi^{a_1}		\& 	\& 	 \&	\& \& w_1 \\
	\& 		\& \& {}^0	\& \& \\
	\& 		\& \&	\&\& \\
		\& \&	 \&	\& \& \\
	\& 		\& \&	\& \pi^{a_s} \& \\
		\& \&		\& 	\& \& \\
		\& 	\&0	\& 	\& \& \\
	\& 		\& 	\&\& \& w_{\delta} \\
} ;
\draw[loosely dotted] (m-1-1)-- (m-5-5);
\draw[loosely dotted] (m-1-6)-- (m-8-6);
\draw (m-2-1)-- (m-6-5) -- (m-8-5) -- (m-8-1) -- (m-2-1);
\draw (m-1-2)-- (m-1-5) -- (m-4-5) -- (m-1-2);
\end{tikzpicture}.\]
\normalsize
In this setting, $c = a_s$.

Moreover, we can deduce from this equality that $d'_{s+1}$ is of the form $a_1+\dots+a_s+val(w_k)$ for some $k>s$. Indeed, the non-zero minors $(s+1) \times (s+1)$ of $PM'Q'$ are all of the following form: they correspond to the choice of $(s+1)$ row linearly independent, and all the rows of index at least $(s+1)$ are in the same one-dimensional sub-space. With such a choice of rows, the corresponding minor is the determinant of a triangular matrix, whose diagonal coefficients are $\pi^{a_1},\dots,\pi^{a_s},w_k$.

On the other hand, $a_1+\dots+a_{s-1}+val(w_k)$ is the valuation of an $s \times s$ minor of $PM'Q'$. By definition, we then have $d'_s \leq a_1+\dots+a_{s-1}+val(w_k)$. Since $d_{s+1}'=a_1+\dots+a_s+val(w_k)$ and $c'=d_{s+1}'-d_{s}'$, we deduce that $c' \geq a_s= c,$ \textit{q.e.d.}.
\end{proof}

We introduce the following notation:

\begin{defn}
Let $E$ be an $R$-module and $X \subset E$ a finite subset. We write $Vect_R(X)$ for the $R$-module generated by the vectors of $X$.
\end{defn}

The previous lemma has then the following consequence:

\begin{lem} \label{lem ecriture avec snf dans fglm}
Let $I,G_1,\leq,\leq_2,B_\leq,B_{\leq_2}$ be as in Theorem \ref{thm stabilite fglm}. Let $x^\beta \in \mathscr{B}_{\leq_2}(I)$. Let $V=Vect_R ( \{ NF_\leq (x^\alpha) \vert x^\alpha \in B_{\leq_2}, \:  x^\alpha < x^\beta \} )$ then $NF_\leq (x^\beta) \in \pi^{-cond_{\leq, \leq_2}(I)} V$
\end{lem}
\begin{proof}
The proof of the correction of the classical FGLM algorithm shows that, if $\mathbf{v}$ is a matrix whose columns are the $NF_\leq (x^\alpha)$ with $x^\alpha \in B_{\leq_2}$ and $ x^\alpha < x^\beta$ (written in the basis $B_\leq$), then  $NF_\leq (x^\beta) \in Im(\mathbf{v}).$

By applying the proof of the Proposition \ref{prop sys lin pas carre}, we obtain that $NF_\leq (x^\beta) \in \pi^{-cond(\mathbf{v})} Vect_R ( \{ NF_\leq (x^\alpha) \vert x^\alpha \in B_{\leq_2}, \:  x^\alpha < x^\beta \} )$.

Finally, Lemma \ref{lem snf iteree} implies that $cond(\mathbf{v}) \leq cond_{\leq, \leq_2}(I).$ The result is then clear.
\end{proof}

\subsubsection{Correction and termination}

We can now prove the correction and termination of Algorithm \ref{FGLM stabilise} under the assumption that the initial precision is enough. Which precision is indeed enough is addressed in the following Subsubsection.

\begin{prop} \label{prop correction et terminaison pour FGLM stabilise}
Let $G_1,\leq,\leq_2,B_\leq,B_{\leq_2}, I$ be as in Theorem \ref{thm stabilite fglm}.
Then, assuming that the coefficients of the polynomials of $G_1$ are all known up to a precision $O(\pi^N)$ for some $N \in \mathbb{Z}_{>0}$ big enough, the stabilized FGLM algorithm \ref{FGLM stabilise} terminates and returns a Gröbner basis $G_2$ of $I$ for $\leq_2$.
\end{prop}
\begin{proof}
The computation of the multiplication matrices only involves multiplication and addition and the operation performed do not depend on the precision. This is similar for the computation of the $NF_{\leq}(x^\alpha)$ processed in the algorithm and obtained as product of $T_i$'s and $\mathbf{1}$. We may assume that all those $NF_{\leq}(x^\alpha)$, for $\vert x^\alpha \vert$, are obtained up to some precision $O(\pi^N)$ for some $N \in \mathbb{Z}_{>0}$ big enough. Subsubsection \ref{subsubsec:analysis-loss-fglm} gives a precise estimation on such an $N$ and when it is big enough. 

Let $M$ be the matrix whose columns are the $NF_\leq (x^\beta)$ for $x^\beta \in  B_{\leq_2}$. Let $cond_{\leq, \leq_2}$ be as in Definition \ref{defn cond FGLM}.

To show the result, we use the following loop invariant: at the beginning of each time in the \textbf{while} loop of Algorithm \ref{FGLM stabilise}, we have \textit{(i)}, $B_2 \subset B_{\leq_2}$ and \textit{(ii)} if $x^\beta = B_2[j] x_i$ (where $(j,i)=m$, $m$ taken at the beginning of the loop), then every monomial $x^\alpha <_2 x^\beta$ satisfies $x^\alpha \in B_{\leq_2}$ or $NF_{\leq}(x^\alpha) \in \pi^{N-cond_{\leq, \leq_2}} Vect_R(NF_\leq (B_{\leq_2}))+O(\pi^{N-cond_{\leq, \leq_2}}).$ Here, \penalty-3000 $O(\pi^{N-cond_{\leq, \leq_2}})$ is the $R$-module generated by the \penalty-30000 $\pi^{N-cond_{\leq, \leq_2}} \epsilon$'s for $\epsilon \in B_\leq.$

We begin by first proving that this proposition is a loop invariant.
It is indeed true when entering the first loop since $ 1 \in B_{\leq_2},$ for $I$ is zero-dimensional.

We then show that this proposition is stable when passing through a loop. Let $x^\beta = B_2[j] x_i$ with $(j,i)=m$.
By the way we defined it, $x^\beta $ is in the border of $B_2$ (\textit{i.e.} non-trivial multiple of a monomial of $B_2$). Since $B_2 \subset B_{\leq_2}$, we deduce that $x^\beta $ is also in either in $B_{\leq_2}$, or in the border of $B_{\leq_2}$, also denoted by $\mathscr{B}_{\leq_2}(I)$.

We begin by the second case. We then have, thanks to Lemma \ref{lem ecriture avec snf dans fglm},  $NF_\leq (x^\beta) \in \pi^{-cond_{\leq, \leq_2}} Vect_R ( \{ NF_\leq (x^\alpha) \vert x^\alpha \in B_{\leq_2}, \:  x^\alpha < x^\beta \} )$. Precision being finite, it tells us that $\lambda = P_1 v=P_1 NF_\leq (x^\beta)$ only appears with coefficients of the form $O(\pi^{l'})$ for its coefficients or row of index $i>s$. This corresponds to being in the image of $\Delta$.

Hence, the \textbf{if} test succeeds, and $x^\beta$ is not added to $B_2$. Points \textit{(i)} and \textit{(ii)} are still satisfied.

We now consider the first case, where $x^\beta \in B_{\leq_2}$. Once again, two cases are possible. The first one is the following: we have enough precision for, when computing $\lambda = P_1 v$ where $v=NF_{\leq}(x^\beta)$, we can prove that $v$ is not in $Vect(NF_\leq (B_{\leq_2}))$. In other words, we are in the \textit{else} case, and $x^\beta$ is rightfully added to $B_2$. The points \textit{(i)} and \textit{(ii)} remain satisfied.
In the other case: we do not have enough precision for, when computing $\lambda = P_1 v$ with $v=NF_{\leq}(x^\beta)$, we can prove that $v$ is not in $Vect(B_{\leq_2})$. In other words, numerically, we get $NF_\leq (x^\beta) \in \pi^{-cond(\mathbf{v})} Vect_R ( \{ NF_\leq (x^\alpha) \vert x^\alpha \in B_{\leq_2}, \:  x^\alpha < x^\beta \} )+O(\pi^{N-cond(\mathbf{v})})$. 
In that case, the \textbf{if} condition is successfully passed and, since $cond(\mathbf{v}) \leq cond_{\leq, \leq_2},$  the points \textit{(i)} and \textit{(ii)} remain satisfied.

This loop invariant is now enough to conclude this demonstration.
Indeed, since $B_2 \subset B_{\leq_2}$ is always satisfied, we can deduce that $L$ is always included in $B_{\leq_2} \cup \mathscr{B}_{\leq_2}(I)$, and since a monomial can not be considered more than once inside the \textbf{while} loop, there is at most $n \delta$ loops. Hence the termination.

Regarding correction, if the \textbf{if} test with $card(B_2)=\delta = card (B_{\leq_2})$ is passed, then, because of the inclusion we have proved, we have $B_2 = B_{\leq_2}$. In that case, the leading monomials which passed the \textbf{if} are necessarily inside the border of $\mathscr{B}_{\leq_2}(I)$, and can indeed be written in the quotient $A/I$ in terms of the monomials of $B_2$ smaller than them. In other words, the linear system solving with the assumption of membership indeed builds a polynomial in $I$. \textit{In fine}, $G_2$ is indeed a Gröbner basis of $I$ for $\leq_2$.

In the second case, where the \textbf{if} test is failed, with \penalty-10000 $card(B_2) \neq \delta$, then precision was not enough.
\end{proof}
\vspace{-0.5cm}
\begin{algorithm}

 \SetKwInOut{Input}{Input}\SetKwInOut{Output}{Output}
\scriptsize

 \Input{$s \in \mathbb{Z}_{\geq 0}$. A matrix $\mathbf{v}$ of size $\delta \times s$, $P_1,Q_1,\Delta$ some matrices such that $P_1 \mathbf{v}' Q_1 = \Delta$ is an \textbf{approximate SNF} of $\mathbf{v}'$ with  $\mathbf{v}'$ the sub-matrix of $\mathbf{v}$ corresponding to its $s-1$ first columns. $P_2,Q_2$ are the inverses of $P_1,Q_1$.}
 \Output{$P_1,P_2,Q_1,Q_2,\Delta$ updated such that $P_1 \mathbf{v} Q_1 = \Delta$ is an approximate SNF of $\mathbf{v}$, and $P_2,Q_2$ are inverses of $P_1,Q_1$. }

Augment trivially the matrices $Q_1,Q_2$ into square invertible matrices with one more row and one more column \;
Compute $U_1,V_1$ and $\Delta'$ realizing an approximate SNF of $P_1 \mathbf{v} Q_1$, as well as $U_2,V_2$ the inverses of $U_1,V_1$ for Algorithm \ref{algo approx snf} \;
$P_1 := U_1 \times P_1$ \;
$Q_1 := Q_1 \times V_1$ \;
$P_2 := P_2 \times U_2$ \;
$Q_2 := V_2 \times Q_2$ \;
$\Delta := \Delta'$ \;

 \caption{Update, iterated approximate SNF} \label{Update FGLM stabilise}
\end{algorithm}
\vspace{-0.5cm}
\begin{algorithm}

 \SetKwInOut{Input}{Input}\SetKwInOut{Output}{Output}
\scriptsize

 \Input{The reduced Gröbner basis $G$ of the zero-dimensional ideal $I \subset A$ for a monomial ordering $\leq$. $\deg I = \delta$. $B_\leq =(1=\epsilon_1 \leq  \epsilon_2 \leq \dots \leq \epsilon_\delta)$ the canonical basis of $A/I$ for $\leq.$ }
 \Output{The multiplication matrices $T_i$'s for $I$  and $\leq$.}

\BlankLine

\For{$i \in \llbracket 1,n \rrbracket$}{
	$T_i:= 0_{\delta \times \delta}$ \;}
$L:= [x_i \epsilon_k \vert i \in \llbracket 1,n \rrbracket \text{ et } \epsilon_k \in B_\leq ],$ ordered increasingly for $\leq$ with no repetition \;
\For{$u \in L$}{
\If{ $u \in \mathscr{E}_\leq (I)$}{$T_i[u,u/x_i]=1$ for all $i$ such that $x_i \vert u$  \;
\tcc{The column indexed by $u$ is zero, except on its coefficient indexed by $u/x_i$ }}
\ElseIf{$u=LM(g)$ for a certain $g \in G$}{
Write $g$ as $u+\sum_{k=1}^\delta a_k \epsilon_k$ \;
$T_i[\cdot,u/x_i]:= -{}^t (a_1,\dots,a_\delta)$ for all $i$ such that $x_i \vert u$ \; }
\Else{
Find the smallest $x_j$ for $\leq$ such that $x_j \vert u.$ \;
Let $v = u/x_j$ \;
Find $\epsilon$ and $l$ such that $v = x_l \epsilon$ \;
$V :=T_l[\cdot , v]$ (this column contains $NF_\leq (v)$) \;
$W := T_j V$ ($W$ is the vector corresponding to the normal form $NF_\leq (x_j v) = NF_\leq (u)$ )\;
$T_i[\cdot, u/x_i]:=W$ for all $i$ such that $x_i \vert u$ \;}

} 
Return $T_1,\dots,T_n$ \;

 \caption{Computation of the multiplication matrices} \label{calcul des matrices de multiplication}
\end{algorithm}
\vspace{-0.2cm}
\subsubsection{Analysis of the loss in precision}
\label{subsubsec:analysis-loss-fglm}


We can now analyse the behaviour of the loss in precision during the execution of the stabilized FGLM algorithm \ref{FGLM stabilise}, and thus estimate what initial precision is big enough for the execution to be without error. 
To that intent, we analyse the precision on the computation of the multiplication matrices and  we use the notion of condition number of Definition \ref{defn cond FGLM} to show that it can handle the behaviour of precision inside the execution the stabilized FGLM algorithm \ref{FGLM stabilise}. This is what is shown in the following propositions. 



\begin{lem}
Let $I,G_1,\leq,B_\leq, \delta, \beta$ be as defined when announcing Theorem \ref{thm stabilite fglm}. Then the coefficients of the multiplication matrices for $I$ are of valuation at least $n \delta \beta.$ \label{lem:val_of_Ti}
\end{lem}
\begin{proof}
$G_1$ is a reduced Gröbner basis of a zero-dimensional ideal. Hence, it is possible to build a Macaulay matrix $\text{Mac}$ with columns indexed by the monomials of $\text{mon} :=\{ X_i \times \epsilon, i \in \llbracket 1, n\rrbracket, \epsilon \in B_\leq   \}$, in decreasing order for $\leq$, and rows of the form $x^\alpha g$, with $x^\alpha$ a monomial and $g \in G$, such that this matrix is under row-echelon form, (left)-injective and all monomials in $\text{mon} \cap LM_{\leq}(I)$ are leading monomial of exactly one row of $\text{Mac}$.
Since $G_1$ is a reduced Gröbner basis, the first non-zero coefficients of the rows are $1$ and all other coefficients are of valuation at least $\beta.$ $\text{Mac}$ has at most $n \delta$ columns and rows.

The computation of the reduced row-echelon form of $\text{Mac}$ yields a matrix whose coefficients are of valuation at least $n \delta \beta$, except the first non-zero coefficient of each row which is equal to $1.$

$NF_\leq (x^\alpha)$ for $x^\alpha \in \text{mon} \setminus B_\leq$ can then be read on the row of $\text{Mac}$ of leading monomial $x^\alpha$. It proves that the coefficients of such a $NF_\leq (x^\alpha)$ are of valuation at least $n \delta \beta.$ The result is then clear.
\end{proof}

\begin{prop}
Let $I,G_1,\leq,\leq_2,B_\leq,B_{\leq_2}$ be as defined when announcing Theorem \ref{thm stabilite fglm}. Let $M$ be the matrix whose columns are the $NF_\leq (x^\beta)$ for $x^\beta \in  B_{\leq_2}$. 
Then, if the coefficients of the polynomials of $G_1$ are all known up to some precision $O(\pi^N)$ with $N \in \mathbb{Z}_{>0}$, $N>cond_{\leq, \leq_2}(M)+n^2(\delta+1)^2 \beta$, the stabilized FGLM algorithm \ref{FGLM stabilise} terminates and returns an approximate  Gröbner basis $G_2$ of $I$ for $\leq_2$. The coefficients of the polynomials of $G_2$ are known up to precision $N-n^2(\delta+1)^2 \beta-2cond_{\leq, \leq_2}(M)$.
\end{prop}
\begin{proof}
We first analyse the behaviour of precision for the computation of the multiplication matrices. There are at most $n\delta$ matrix-vector multiplication in the execution of Algorithm \ref{calcul des matrices de multiplication}. The coefficients involved in those multiplication are of valuation at least $n \delta \beta$ thanks to Lemma \ref{lem:val_of_Ti}. Hence, the coefficients of the $T_i$ are known up to precision $O(\pi^{N-(n \delta)^2 \beta}).$

We now analyse the exection of Algorithm \ref{FGLM stabilise}.
The computation of $v$ involves the multiplication of $\deg (v)$ $T_i$'s and $\mathbf{1}.$
Hence, $v$ is known up to precision $O(\pi^{N-(n \delta)^2 \beta-\deg (v) n \delta \beta}),$ which can be lower-bounded by $O(\pi^{N-(\delta+1)^2 n^2 \beta}).$

As a consequence, all coefficients of $M$ are known up to precision $O(\pi^{N-(\delta+1)^2 n^2 \beta})$ and this is the same for its approximate SNF.

Now, we can address the loss in precision for the linear system solving. Thanks to Proposition \ref{prop sys lin pas carre}, and with the membership assumption of $v$ to $Im(\mathbf{v})$, a precision $O(\pi^N)$ with $N$ strictly bigger than $(\delta+1)^2 n^2 \beta$ plus the biggest valuation $c$ of an invariant factor of  $\mathbf{v}$ is enough to solve the linear system $\mathbf{v}W=v,$ and the coefficients of $W$ are determined up to precision $O(\pi^{N-n^2(\delta+1)^2 \beta-2c})$.
The Lemma \ref{lem snf iteree} then allows us to conclude that at any time, $c \leq cond_{\leq, \leq_2}(M)$, hence the result.
\end{proof}







\subsubsection{Complexity}

To conclude the proof of Theorem \ref{thm stabilite fglm}, what remains is to give an estimation of the complexity of Algorithm \ref{FGLM stabilise}. Regarding to the computation of multiplication matrices, there is no modification concerning complexity, and what we have to study is only the complexity of the iterated SNF computation. This is done in the following lemma:

\begin{lem} \label{lem iteraton snf}
Let $1 \leq s \leq \delta$ and $prec$ be integers, $k \in \llbracket 1, s \rrbracket$ and $M,C^{(k)}$ be two matrices in $M_{\delta \times s}(K).$ We assume that the coefficients of $M$ satisfies $M_{i,j}=m_{i,j} \delta_{i,j}+O(\pi^{prec})$ for some $m_{i,j} \in K$ and the coefficients of $C^{(k)}$ satisfies $C_{i,j}^{(k)}=c_{i,j} \delta_{j,k}+O(\pi^{prec})$ for some $c_{i,j} \in K$.
Let $C_{SNF}(M+C^{(k)})$ be the number of operations in $K$ (at precision $O(\pi^{prec})$) applied on rows and columns to compute an approximate SNF for $M+C^{(k)}$ at precision $O(\pi^{prec}).$ Then $C_{SNF}(M+C^{(k)})\leq s \delta.$
\end{lem}
\begin{proof}
We show this result by induction on $s$.
For $s=1$, for any $\delta, prec,k,M$ and $C^{(k)}$, the result is clear.

Let us assume that for some $s \in \mathbb{Z}_{>0}$, we have for any $\delta, prec, k$, $M$ and $C^{(k)} \in M_{\delta \times (s-1)}(K)$ as in the lemma, $C_{SNF}(M+C^{(k)})\leq (s-1) \delta.$

Then, let us take some $\delta \geq s,$  $k \in \llbracket 1, s \rrbracket$ and $rec \in \mathbb{Z}_{\geq 0}.$ Let $M,C^{(k)}$ be two matrices in $M_{\delta \times s}(K)$ such that their coefficients satisfies $M_{i,j}=m_{i,j} \delta_{i,j}+O(\pi^{prec}),$ for some $m_{i,j} \in K,$ and $C_{i,j}^{(k)}=c_{i,j} \delta_{j,k}+O(\pi^{prec})$ for some $c_{i,j} \in K$. Let $N=M+C^{(k)}.$

We apply Algorithm \ref{algo approx snf} until the recursive call. 
Let us assume that the coefficient used as pivot, that is, one $N_{i,j}$ which attains the minimum of the $val(N_{i,j})$'s,  is $N_{1,1}$. Then $1$ operation on the columns is done when going through the two consecutives \textbf{for} loops in Algorithm \ref{algo approx snf}.
The only other case is that of pivot being some $N_{i,k}$ for some $i$. Then $\delta-1$ operations on the rows and $1$ operation on the columns are done.

The matrix $N'=\widetilde{N}_{i \geq 2, j \geq 2}$ can be written $N'=M'+C^{'(k)}$ with $M'$ and $C^{'(k)}$ in $M_{(\delta-1) \times (s-1)}(K)$ of the desired form, for $k=s-1$ if the pivot $N_{i,j}$ is $N_{1,1}$ and $k=i$ if it is $N_{i,s}$.
By applying the induction hypothesis on $N'$, we obtain that $C_{SNF}(M+C^{(k)})\leq \delta+(\delta-1)\times (s-1) \leq \delta s.$

The result is then proved by induction.
\end{proof}

We then have the following result regarding the complexity of Algorithm \ref{FGLM stabilise} :

\begin{prop} \label{prop complexite fglm stabilise}
Let $G_1$ be an approximate reduced Gröbner basis, for some monomial ordering $\leq,$ of some zero-dimensional $I \subset A$ of degree $\delta$, and let $\leq_2$ be some monomial ordering. We assume that the coefficients of $G_1$ are known up to precision $O(\pi^N)$ for some $N>cond_{\leq,\leq_2}$. Then, the complexity of the execution of Algorithm \ref{FGLM stabilise} is in $O(n \delta^3)$ operations in $K$ at absolute precision $O(\pi^N)$.
\end{prop}
\begin{proof}
Firstly, we remark that the computation of the matrices of multiplication is in $O(n \delta^3)$ operations at precision $O(\pi^N)$.
Now, we consider what happens inside the \textbf{while} loop in Algorithm \ref{FGLM stabilise}. The computation of approximate SNF through Algorithm \ref{Update FGLM stabilise} are in $O(\delta^2)$ operations at precision $O(\pi^N)$ thanks to Lemma \ref{lem iteraton snf}. The solving of linear systems thanks to Proposition \ref{prop sys lin pas carre} are also in $O(\delta^2)$ operations at precision $O(\pi^N)$.
There is at most $n \delta$ entrance in this loop thanks to the proof of termination in Proposition \ref{prop correction et terminaison pour FGLM stabilise}. The result is then proved.
\end{proof}

We can recall that the complexity of the classical FGLM algorithm is also in $O(n \delta^3)$ operations over the base field.

\section{Shape position}
\label{section: stabilite shape}

In this Section, we analyse the special variant of FGLM to compute a shape position Gröbner basis. We show that the gain in complexity observed in the classical case is still satisfied in our setting. We can combine this result with that of \cite{Vaccon:2014} to express the loss in precision to compute a shape position Gröbner basis starting from a regular sequence.

\subsection{Grevlex to shape}
To fasten the computation of the multiplication matrices, we use the following notion.
\begin{defn}
$I$ is said to be semi-stable for $x_n$ if for all $x^\alpha$ such that $x^\alpha \in LM(I)$ and $x_n \mid x^\alpha$ we have for all $k \in \llbracket 1, n-1 \rrbracket$ $\frac{x_k}{x_n} x^\alpha \in LM(I).$
\end{defn}

Semi-stability's application is then explained in Proposition 4.15, Theorem 4.16 and Corollary 4.19 of \cite{Huot:13} (see also \cite{Faugere-Huot:2013}) that we recall here:
\begin{prop} \label{prop:Huot} Applying FGLM for a zero-dimensional ideal $I$ starting from a Gröbner basis $G$ of $I$ for grevlex:\\
  \texttt{1.} $T_i 1$ ($i<n$) can be read from $G$ and requires no arithmetic operation;\\
  \texttt{2.} If $I$ is semi-stable for $x_n,$ $T_n$ can be read from $G$ and requires no arithmetic operation;\\
  \texttt{3.} After a generic change of variable, $I$ is semi-stable for $x_n.$
\end{prop}

The FGLM algorithm can then be adapted to this setting in the special case of the computation of a 
Gröbner basis of an ideal in shape position, with Algorithm \ref{FGLM stabilise avec shape depuis grevlex}.

\begin{rmk}
If the ideal $I$ is weakly grevlex (or the initial polynomials satisfy the more restrictive \textbf{H2} of \cite{Vaccon:2014}), then $I$ is semi-stable for $x_n.$
\end{rmk}

The remaining of this Section is then devoted to the proof of Theorem \ref{thm:FGLMp:grevlex:lex}.

%
%

\subsection{Correction, termination and precision}

We begin by proving correction and termination of this algorithm.

\begin{prop}
We assume that the coefficients of the polynomials of the reduced Gröbner basis $G_1$ for grevlex are known up to a big enough precision, and that the ideal $I=\left\langle G_1 \right\rangle$ is in general position and semi-stable for $x_n$. Then Algorithm \ref{FGLM stabilise avec shape depuis grevlex}  terminates and returns a Gröbner basis for lex of $I$, yielding an univariate representation. Time complexity is in $O(\delta^3) + O(n \delta^2)$.
\end{prop}
\begin{proof}
As soon as one can certify that the rank of $M$ is $\delta$, the dimension of $A/I$, then we can certify that $I$ possesses an univariate representation.
Correction, termination are then clear. Computing $T_n$ and the $T_i 1$ is free, computing the SNF is in $O(\delta^3)$ and solving the linear systems is in $O(n \delta^2),$ hence the complexity is clear.
\end{proof}

What remains to be analysed is the loss in precision. To that intent, we use again the condition number of $I$ (from grevlex to lex) and the smallest valuation of a coefficient of $G_1.$

\begin{prop} \label{prop: prec of shape from grevlex}
Let $G_1$ be the reduced Gröbner basis for grevlex of some zero-dimensional ideal $I \subset A$ of degree $\delta$. We assume that the coefficients of the polynomials of $G_1$ are known up to precision $O(\pi^N)$ for some $N \in \mathbb{Z}_{>0},$ except the leading coefficients, which are exactly equal to $1$. Let $\beta$ be the smallest valuation of a coefficient of $G_1.$
Let $m=cond_{grevlex, lex}(I).$ We assume that $m-\delta \beta <N$, that $I$ is in shape position and semi-stable for $x_n$.
Then Algorithm \ref{FGLM stabilise avec shape depuis grevlex} computes a Gröbner basis $(x_1-h_1,\dots,x_{n-1}-h_{n-1}, h_n)$ of $I$ for lex which is in shape position. Its coefficients are known up to precision $O(\pi^{N-2m+\delta \beta})$. The valuation of the coefficients of $h_n$ is at least $\beta \delta-m,$ and those of the $h_i$'s is at least $\beta-m.$
\end{prop}
\begin{proof}
There is no loss in precision for the computation concerning the multiplication
matrices since it only involves reading coefficients on $G_1$. Their
coefficients are of valuation at least $\beta.$ The columns of
$M:=Mat_{B_{grevlex}}(NF_\leq(1),\dots,$ $NF_\leq (x_n^{\delta-1}))$ are obtained
using $T_n.$ Their coefficients are known up to precision $O(\pi^{N+(\delta-1)
  \beta})$ and are of valuation at least $(\delta-1) \beta.$ For
$\mathbf{z}[\delta]$, it is $O(\pi^{N+\delta \beta})$ and $\delta \beta.$ The
only remaining step to analyse is then the solving of linear systems, which is
clear thanks to Theorem \ref{thm sys lin}.
\end{proof}

\subsection{Summary on shape position}

Thanks to the results of \cite{Vaccon:2014} and \cite{Vaccon-these},
 we can express the loss in precision to compute a Gröbner basis
 in shape position
 under some genericity assumptions.
Let $F=(f_1,\dots,f_n) \in R[X_1,\dots, X_n]$ be a sequence of polynomials 
satisfying the hypotheses \textbf{H1} and \textbf{H2} of \cite{Vaccon:2014} for grevlex.
Let $D$ be the Macaulay bound of $F$ and $I=\left\langle F \right\rangle.$ 
We assume that $I$ is strongly stable for $x_n.$ Let $\delta = \deg (I).$
Let $\beta=-prec_{MF5}(F, D, grevlex)$ be the bound on loss in precision to compute 
an approximate 
grevlex Gröbner basis of \cite{Vaccon:2014}.
Let $\gamma=-\delta \beta+2 cond_{grevlex,lex}(I)$.

\begin{theo}
If the coefficients of the $f_i$'s are known up to precision $N>\gamma$, then 
one can compute a shape position Gröbner basis for $I$ with precision $N-\gamma$ on its coefficients.
\end{theo}
\begin{proof}
An approximate reduced Gröbner basis of $I$ for grevlex is determined up to precision $N+2\beta$ and its coefficients are of valuation at least $\beta.$
Thanks to Proposition \ref{prop: prec of shape from grevlex}, the lexicographical Gröbner basis of $I$ is of the form $x_1-h_1(x_n),\dots,$ $x_{n-1}-h_{n-1}(x_n),h_n(x_n).$ Moreover, the coefficients of $h_n$ are of valuation at least $\delta \beta-cond_{grevlex,lex}(I)$ and known at precision
$N-\delta \beta-2cond_{grevlex,lex}(I).$
For the other $h_i$'s, the coefficients are of valuation at least $\beta-cond_{grevlex,lex}(I)$ and precision
$N-\gamma.$
%
\end{proof}
\begin{rmk}
As a corollary, if $x_n \in R$ is such that  $val(f_n'(x))$ $=0$, then $x_n$ lifts to $x \in V(I),$ known at precision $N-2\gamma.$
\end{rmk}

\begin{algorithm}

 \SetKwInOut{Input}{Input}\SetKwInOut{Output}{Output}
 \scriptsize

 \Input{An approximate reduced Gröbner basis $G_1$ for grevlex of some ideal $I \subset A$ of  dimension zero and degree $\delta$. $I$ is semi-stable for $x_n$ and in in shape position.}
 \Output{An approximate Gröbner basis $G_2$ of $I$ for $\leq_{lex}$, in shape position, or \textbf{Error} if the precision is not enough.}

Read the multiplication matrix $T_n$ for $I$ and grevlex using $G$\;
$G_2 := \emptyset$ \;
Read the $\mathbf{y}[i]:= T_i 1$'s from $G$ ($1 \leq i <n$) \;
$\mathbf{z}[0]:=1$ \;
\For{$i$ from $1$ to $\delta$}{
Compute $\mathbf{z}[i]= T_n \mathbf{z}[i-1]$ \;}
$M:= Mat_{B_\leq}(\mathbf{z}[0],\dots,\mathbf{z}[\delta-1])$ \;
Compute $\Delta$ the SNF of $M$ with $\Delta = P M Q$ \;
\eIf{ $rank(M)==\delta$}{
\For{$i$ from $1$ to $n-1$}{
Let $U$ s.t. $\mathbf{y}[i]=-M \cdot U$ thanks to $P,Q,\Delta$ and Thm \ref{thm sys lin} \;
 $h_i(T):=\sum_{i=1}^{\delta-1} U[i] T^i $ \;}
Let $U$ s.t. $\mathbf{z}[\delta]=-M \cdot U$ thanks to $P,Q,\Delta$ and Thm \ref{thm sys lin} \;
 $h_n(T):=T^\delta+\sum_{i=1}^{\delta-1} U[i] T^i $ \;
 Return $x_1-h_1(x_n),\dots,x_{n-1}-h_{n-1}(x_n),h_n(x_n)$ \;}
{Return "\textbf{Error}, not enough precision"} 
 \caption{Stabilized FGLM algorithm for an ideal in shape postition starting from grevlex} \label{FGLM stabilise avec shape depuis grevlex}
\end{algorithm}
\vspace{-0.5cm}
\section{Experimental Results}
An implementation in Sage \cite{sage} of the previous algorithms is
available at
\url{http://www2.rikkyo.ac.jp/web/vaccon/fglm.sage}. Since the main
goal of this implementation is the study of precision, it has not been
optimized regarding to time-complexity.  We have applied the main
Matrix-F5 algorithm of \cite{Vaccon:2014} to homogeneous polynomials
of given degrees, with coefficients taken randomly in $\mathbb{Z}_p$
(using the natural Haar measure): $f_1,\dots, f_s,$ of degree $d_1,
\dots, d_s$ in $\mathbb{Z}_p[X_1,\dots, X_s],$ known at precision
$O(p^{150}),$ for grevlex, using the Macaulay bound $D$. We also used
the extension to the affine case of \cite{Vaccon:2014} to handle
affine polynomials with the same setting (we specify this property in
the column aff.). We have then applied our p-adic variant of FGLM
algorithm, specialized for grevlex to lex or not
, on the obtained Gröbner bases to get Gröbner bases for the lex
order.\\
\begin{scriptsize}
\begin{tabular}{|c|c|c|c|c|c|c|c|c|c|}
\hline
$d =$ & $nb_{test}$ & aff. & fast
& D & $p$ 
& max & mean & fail \\  \hline

[3,3,3]& 20 & no  & no 
& 7 &2
& 21 & 3 & (0,0)\\ \hline
[3,3,4]& 20 & no  & no 
& 8 &2
& 21 & 3 & (0,0)\\ \hline
[4,4,4]& 20 & no  & no 
& 10 &2
& 28 & 5.2 & (0,0)\\ \hline \hline
[3,3,3]& 20 & yes  & no 
& 7 &2
& 150 & 78 & (0,0)\\ \hline
[3,3,4]& 20 & yes  & no 
& 8 &2
& 149 & 92 & (0,5)\\ \hline
[4,4,4]& 20 & yes  & no 
& 10 &2
& 150 & 118 & (0,11)\\ \hline \hline
[3,3,3]& 20 & yes  & yes 
& 7 &2
& 145 & 65 & (0,1)\\ \hline
[3,3,4]& 20 & yes  & yes 
& 8 &2
& 150 & 89 & (0,7)\\ \hline
[4,4,4]& 20 & yes  & yes 
& 10 &2
& 156 & 124 & (0,15)\\ \hline \hline \hline \hline 
[3,3,3]& 20 & no  & no 
& 7 &65519
& 0 & 0 & (0,0)\\ \hline
[4,4,4]& 20 & no  & no 
& 10 &65519
& 0 & 0 & (0,0)\\ \hline \hline
[3,3,3]& 20 & yes  & no 
& 7 &65519
& 0 & 0 & (0,0)\\ \hline
[4,4,4]& 20 & yes  & no 
& 10 &65519
& 0 & 0 & (0,0)\\ \hline \hline
[3,3,3]& 20 & yes  & yes 
& 7 &65519
& 0 & 0 & (0,0)\\ \hline
[4,4,4]& 20 & yes  & yes 
& 10 &65519
& 0 & 0 & (0,0)\\ \hline
\end{tabular} 
\end{scriptsize}

This experiment has been realized $nb_{test}$ times for each given
choice of parameters. We have reported in the previous array the maximal
(column max), resp. mean (column mean), loss in precision (in
successful computations), and the number of failures. This last
quantity is given as a couple: the first part is the number of failure
for the Matrix-F5 part and the second for the FGLM part.

We remark that these results suggest a difference of order in the loss
in precision between the affine and the homogeneous case. 
Qualitatively, we remark that, for some given initial degrees,
more computation (particularly computation involving loss in
precision) are done in the affine case, because of the inter-reduction
step. Also, it seems clear that loss in precision decreases when $p$
increases, in particular, on small instances like here, loss in
precision when $p=65519$ are very unlikely.
\vspace{-0.3cm}
\section{Future works}
Following this work, it would be interesting to investigate whether the
sub-cubics algorithms of 
\cite{Faugere-Mou:2011,Faugere-Mou:2013,Mou:2013,Faugere-Huot:2013,Faugere:2014,Huot:13} 
could be adapted to the $p$-adic setting with reasonable loss in precision.
Another possibility of interest for $p$-adic computation would be the 
extension of FGLM to tropical Gröbner bases.
\vspace{-0.3cm}

\scriptsize

\bibliographystyle{alpha}
\bibliography{biblio}

\newcommand{\etalchar}[1]{$^{#1}$}
\begin{thebibliography}{FGLM93}

\bibitem[CL14]{Caruso:2014}
Xavier Caruso and David Lubicz.
\newblock Linear algebra over {$\Bbb{Z}_p[[u]]$} and related rings.
\newblock {\em LMS J. Comput. Math.}, 17(1):302--344, 2014.

\bibitem[CRV14]{Caruso:2014:2}
Xavier Caruso, David Roe, and Tristan Vaccon.
\newblock Tracking {$p$}-adic precision.
\newblock {\em LMS J. Comput. Math.}, 17(suppl. A):274--294, 2014.

\bibitem[CRV15]{Caruso:2015}
Xavier Caruso, David Roe, and Tristan Vaccon.
\newblock {p-Adic Stability In Linear Algebra}.
\newblock pages 101--108, 2015.

\bibitem[FGHR13]{Faugere-Huot:2013}
Jean-Charles Faugère, Pierrick Gaudry, Louise Huot, and Guénaël Renault.
\newblock {Polynomial Systems Solving by Fast Linear Algebra}.
\newblock preprint, 2013.
\newblock 23 pages.

\bibitem[FGHR14]{Faugere:2014}
Jean-Charles Faugère, Pierrick Gaudry, Louise Huot, and Guénaël Renault.
\newblock {Sub-cubic Change of Ordering for Gröbner Basis: A Probabilistic
  Approach}.
\newblock In {\em {Proceedings of the 39th International Symposium on Symbolic
  and Algebraic Computation}}, pages 170--177, Kobe, Japon, July 2014. ACM.

\bibitem[FGLM93]{Faugere:1993}
Jean-Charles Faugère, Patrizia Gianni, Daniel Lazard, and Teo Mora.
\newblock {Efficient computation of zero-dimensional Gr{\"o}bner bases by
  change of ordering}.
\newblock {\em Journal of Symbolic Computation}, 16(4):329--344, 1993.

\bibitem[FM11]{Faugere-Mou:2011}
Jean-Charles Faugère and Chenqi Mou.
\newblock {Fast Algorithm for Change of Ordering of Zero-dimensional Gröbner
  Bases with Sparse Multiplication Matrices}.
\newblock In {\em Proceedings of the 36th international symposium on Symbolic
  and algebraic computation}, ISSAC '11, pages 115--122, New York, NY, USA,
  2011. ACM.

\bibitem[FM13]{Faugere-Mou:2013}
Jean{-}Charles Faug{\`{e}}re and Chenqi Mou.
\newblock Sparse {FGLM} algorithms.
\newblock {\em CoRR}, abs/1304.1238, 2013.

\bibitem[Huo13]{Huot:13}
Louise Huot.
\newblock {\em Résolution de systèmes polynomiaux et cryptologie sur les
  courbes elliptiques}.
\newblock PhD thesis, Université Pierre et Marie Curie (Paris VI), December
  2013.
\newblock \url{http://tel.archives-ouvertes.fr/tel-00925271}.

\bibitem[Mou13]{Mou:2013}
Chenqi Mou.
\newblock {\em {Solving Polynomial Systems over Finite Fields: Algorithms,
  Implementation and Applications}}.
\newblock Theses, {Universit{\'e} Pierre et Marie Curie}, May 2013.

\bibitem[S{\etalchar{+}}11]{sage}
W.\thinspace{}A. Stein et~al.
\newblock {\em {S}age {M}athematics {S}oftware ({V}ersion 4.7.2)}.
\newblock The Sage Development Team, 2011.
\newblock {\tt http://www.sagemath.org}.

\bibitem[Ser79]{Serre:1979}
Jean-Pierre Serre.
\newblock {\em Local fields}, volume~67 of {\em Graduate Texts in Mathematics}.
\newblock Springer-Verlag, New York-Berlin, 1979.
\newblock Translated from the French by Marvin Jay Greenberg.

\bibitem[TW95]{Taylor:1995}
Richard Taylor and Andrew Wiles.
\newblock Ring-theoretic properties of certain {H}ecke algebras.
\newblock {\em Ann. of Math. (2)}, 141(3):553--572, 1995.

\bibitem[Vac14]{Vaccon:2014}
Tristan Vaccon.
\newblock {Matrix-F5 algorithms over finite-precision complete discrete
  valuation fields}.
\newblock In {\em {Proceedings of the 2014 {ACM} on International Symposium on
  Symbolic and Algebraic Computation, {ISSAC} '14, Kobe, Japan, July 23-25,
  2014}}, pages 397--404, 2014.

\bibitem[Vac15]{Vaccon-these}
Tristan Vaccon.
\newblock {\em {p-adic precision}}.
\newblock Theses, {Universit{\'e} Rennes 1}, July 2015.

\end{thebibliography}

\end{document}